\documentclass[10pt]{article}
\usepackage{amsmath,amsfonts,xypic}
\usepackage{amssymb}
\usepackage{amscd}
\usepackage[colorlinks=true,urlcolor=blue]{hyperref}
\hypersetup{				
backref = true,				
pagebackref = true,			
hyperindex = true, 			
colorlinks = true, 			
breaklinks = true, 			
urlcolor = blue, 			
linkcolor = blue, 			
bookmarks = true,			
bookmarksopen = true, 		
}
\advance \textheight by \topskip
\makeatletter
\@addtoreset{equation}{section}
 
\makeatother

\newtheorem{theorem}{Theorem}[section]

\newtheorem{definition}[theorem]{Definition}
\newtheorem{example}[theorem]{Example}

\newtheorem{lemma}[theorem]{Lemma}

\newtheorem{proposition}[theorem]{Proposition}
\newtheorem{remark}[theorem]{Remark}

\newenvironment{proof}[1][Proof]{\textbf{#1.} }{\ \rule{0.5em}{0.5em}}

\def\gr{\text{gr}}
\def\l{\left[\left|}
\def\r{\right|\right]_{_{\hskip -.05 truecm N}}}
\def\lb{\left\{\left|}
\def\rb{\right|\right\}_{_{\hskip -.05 truecm N}}}

\def\ZZ{{\mathbb Z}}

\newcommand{\beq}{\begin{equation}}
\newcommand{\eeq}{\end{equation}}
\newcommand{\beqa}{\begin{eqnarray}}
\newcommand{\eeqa}{\end{eqnarray}}
\newcommand{\noi}{\noindent}

\newcommand{\g}{{\mathfrak g}}

\def\>{\rangle}
\def\<{\langle}

\begin{document}

\title{
{\bf  Color Lie algebras and Lie algebras of order $F$ }  }

\author{
{\sf R. Campoamor-Stursberg}
\thanks{e-mail:rutwig@pdi.ucm.es}$\,\,$${}^{a}$,
 {\sf M. Rausch de Traubenberg }\thanks{e-mail:
rausch@lpt1.u-strasbg.fr}$\,\,$${}^{b}$
\\
{\small ${}^{a}${\it Instituto de Matem\'atica Interdisciplinar,}}\\
{\small {\it Universidad Complutense de Madrid, 3 Plaza de
Ciencias, 28040 Madrid, Spain.}}\\
  \\
{\small ${}^{b}${\it Laboratoire de Physique Th\'eorique, CNRS UMR
7085,
Universit\'e Louis Pasteur}}\\
{\small {\it  3 rue de
l'Universit\'e, 67084 Strasbourg Cedex, France}}}
\date{\today}
\maketitle
\vskip-1.5cm

\vspace{2truecm}

\begin{abstract}
The notion of color algebras is generalized to the class of
$F$-ary algebras, and corresponding decoloration theorems are
established. This is used to give a construction of colored
structures by means of tensor products with Clifford-like
algebras. It is moreover shown that color algebras admit
realisations as $q=0$ quon algebras.
\par\smallskip
{\bf 2000 MSC:} 17B70, 17B75, 17B81
\end{abstract}

\section{Introduction}

 The problem of finding mass formulae for particles belonging to a
representation of an interaction group motivated, in the beginning
sixties, the efforts to combine interactions with relativistic
invariance in a non-trivial way. The well known obstructions to
such a construction \cite{Cole} finally led to the supersymmetric
schemes. In this sense, two different models unifying internal and
space-time symmetries, the conformal superalgebras
$\frak{su}(2,2;N)$ and the orthosymplectic algebra
$\frak{osp}(4;N)$ were proposed \cite{Haa}. The introduction of a
grading was soon recognized to be an indispensable requirement to
introduce transformations relating states obeying different
quantum statistics types. However, a further generalization seemed
necessary to clearly distinguish the color and flavor degrees of
freedom, which were treated in the same way in the two previous
models. A first approach to this question was made in \cite{Gun},
where a color algebra based on the nonassociative octonions was
proposed. This scheme constituted a mathematical model taking into
account the unobservability of quarks and their associated
massless color gauge bosons. An alternative construction,
preserving the associative framework, was given in \cite{cla1}.
These structures motivated by themselves the study of several
extensions of Lie algebras, keeping in mind the main properties
that made them interesting to describe symmetries of physical
phenomena. Among others, the generalizations that have been proven
to be physically relevant are color (or graded) Lie algebras
\cite{sylv,gj,cla-phys,cla1,sch} and, more recently, Lie algebras of
order $F$ \cite{flie,flie-phys}. These two types of algebras share
some properties, and are based upon a grading by an Abelian group.

\medskip

Tensor products constitute a natural tool to construct higher
dimensional algebraic objects starting from two given ones, as
well as to study their representation theory and the underlying
Clebsch-Gordan problem. However, the tensor product of two
algebras usually give rise to non-trivial identities that must be
satisfied, and often fail to preserve certain key properties (as
happens e.g. for Lie algebras). In order to prevent this
situation, generalized tensor products have been developed for
various structures, such as groups or Lie (super)algebras. In a
more general frame, there is no reason to believe that tensor
products of algebras with different structures can not lead to
further interesting structures, possibly preserving some of the
main properties of its components. Some attention has been
devoted, in this direction, to tensor products of the type
$C\otimes D$, where $C$ is a Clifford algebra $C$ and $D$ a
$\mathbb{Z}_{2}$-graded ring of differential operators on a
manifold, where the differential operators are interpreted, in
some sense, as ``quantum mechanical", the classical approximation
of which is given by a Poisson bracket. These products suggest a
deep relation with the commutation-anticommutation formalism in
Field Theory \cite{Nee}.

\medskip

In this paper we show that color algebras and Lie algebras of order
$F$  can
be unified leading to some new algebras that we call color Lie
(super)algebra of order $F$. Furthermore, we show that many
examples of these algebras can be seen as tensor product of given
algebras. In Section 2 we recall some results concerning the
general theory of color algebras, focusing on the isomorphism
between color and graded algebras \cite{ba}. We also review the
main features of a distinguished class of $F$-ary algebra. Section
4 is devoted to give a unification of both mentioned types of
structures, as well as an adapted decoloration theorem. It turns
out that color algebras arise as tensor products of ordinary
non-color algebras with algebras of Clifford type. In the last
section we show that all considered types of algebras are strongly
related to quon algebras for $q=0$. In particular, ``differential"
realisations in terms of quon algebras are obtained.

 \section{Color Lie algebras}\label{color}

Color Lie (super)algebras, originally introduced in \cite{cla1},
can be seen as a direct generalization of Lie (super)algebras.
Indeed, the latter are defined through antisymmetric (commutator)
or symmetric (anticommutator) products, although for the former
the product is neither symmetric nor antisymmetric and is defined
by means of a commutation factor. This commutation factor is equal
to $\pm 1$ for  (super)Lie algebras and   more  general for
 arbitrary color Lie (super)algebras.

As happened for Lie superalgebras, the basic tool to define color
Lie (super)algebras is a grading determined by an Abelian group.
The latter, besides defining the underlying grading in the
structure, moreover provides a new object known as commutation
factor.

\begin{definition}\label{abelian}
Let $\Gamma$ be an abelian group. A commutation factor $N$ is a
map $N: \ \Gamma \times \Gamma \to \mathbb
C\setminus\left\{0\right\}$
 satisfying the
following constraints:

\begin{enumerate}
\item  $N(a,b) N(b,a)=1, $ for all $a,b \in \Gamma$; \item
$N(a,b+c)= N(a,b) N(a,c),$ for all  $a,b,c  \in \Gamma$; \item
$N(a+b,c)= N(a,c) N(b,c),$ for all $a,b,c  \in \Gamma$.
\end{enumerate}
\end{definition}

The definition above implies, in particular, the following
relations
 \beqa N(0,a)=N(a,0)= 1,\quad N(a,b)=N(-b,a),\quad N(a,a)= \pm 1\
\text{ for all } a,b \in \Gamma, \eeqa where $0$ denotes the
identity element of $\Gamma$. In particular, fixed one element of
$\Gamma$, the induced mapping $N_{x}: \Gamma \to \mathbb C
\setminus\left\{0\right\}$ defines a homomorphism of groups.

\begin{definition}\label{colored}
Let $\Gamma$ be an abelian group and $N$ be a commutation factor.
The (complex) graded vector space $\g = \underset{a \in
\Gamma}{\oplus}\g_a$ is called a color Lie (super)algebra if
\begin{enumerate}
\item $\g_0$ is a (complex) Lie algebra.
 \item For all $a \in
\Gamma\setminus\left\{0\right\}$, $\g_a$ is a representation of
$\g_0$. If $X \in \g_0, Y \in \g_1$ then $\l X,Y \r=[X,Y]$ denotes the
action of $X$ on $Y$.
 \item For all $a, b \in \Gamma$, there exists a
$\g_0-$equivariant map $ \l \ , \ \r \g_a \times \g_b \to
\g_{a+b}$ such that for all $X \in \g_a, Y \in \g_b$  the
constraint $\l X,Y\r = -N(a,b) \l Y,X\r$ is satisfied. \item For
all $X \in \g_a,Y \in \g_b, Z \in \g_c$, the following Jacobi
identities

$$\l X,\l Y,Z \r \r = \l \l X, Y \r, Z \r + N(a,b) \l Y, \l X ,Z \r \r.$$
hold.
\end{enumerate}
\end{definition}

\begin{remark}
The Jacobi identity above can be rewritten in equivalent form as
$$N(c,a) \l X,\l Y,Z \r \r + N(a,b) \l Y,\l Z,X \r \r +
 N(c,a) \l Z,\l Y,X \r \r=0.$$
Further, property 1 in definition  \ref{abelian} is a consequence
of 3 in definition \ref{colored}, while  property 2 in Definition
\ref{abelian} is a consequence of the Jacobi identity 4 in
Definition \ref{colored}.  For the particular case
$\Gamma=\left\{0\right\}$, $\g=\g_0$ reduces to a Lie algebra. If
 $\Gamma=\mathbb Z_2$, we obtain the grading $\g=\g_0 \oplus \g_1$. If in
addition   $N(1,1)=-1$ holds, $\g$ is just a Lie superalgebra.
Therefore, the latter condition in definition \ref{colored} points
out to which extent color algebras extend ordinary Lie algebras
and superalgebras. Furthermore, if $N(a,b)=1$ for all $a,b \in
\Gamma$, $\g$ is a Lie algebra graded by the group  $\Gamma$. From
now Lie algebras with this last property will be called
$\Gamma-$graded Lie algebras.
\end{remark}

The grading group $\Gamma$ inherits naturally a $\mathbb
Z_2-$grading: $\Gamma=\Gamma_0 \oplus \Gamma_1$, where
$\Gamma_i=\Big\{a \in \Gamma |$  $
N(a,a)=(-1)^i\Big\}$ \cite{sch}. If $\Gamma_1 = 0$ ($\Gamma_1
\ne 0$), $\g$ is called a color Lie algebra (respectively
superalgebra). Starting from a color Lie superalgebra, we define
$N_+$ by $N_+(a,b)=(-1)^{|a| |b|}N(a,b)$, where $|a|$ is the
degree of $a$ with respect to the $\mathbb Z_2-$grading. It is not
difficult to check that $N_+$ is also a commutation factor.
Furthermore, if we decompose $\g=\g_0 \oplus \g_1$ with respect to
this $\mathbb Z_2-$grading, and introduce the Grassmann algebra
$\Lambda(\mathbb C^m)$, the analogue of the Grassmann-hull in the
case of Lie superalgebras, we can endow the color Lie superalgebra
with a color Lie algebra structure. Indeed, if we set
$\Lambda(\mathbb C^m)= \Lambda(\mathbb C^m)_0 \oplus
\Lambda(\mathbb C^m)_1$, where $\Lambda(\mathbb C^m)_i$ are of
degree $i$ then

$$
(\Lambda(\mathbb C^m) \otimes \g)_0= \Lambda(\mathbb C^m)_0
\otimes \g_0 \oplus  \Lambda(\mathbb C^m)_1 \otimes \g_1,
$$

\noi is a color Lie algebra with commutation factor $N_+$.

\begin{remark}
\label{assos-color} To any associative $\Gamma-$graded algebras
 ${\cal A}=\underset{a \in \Gamma} {\oplus}
 {\cal A}_a$ with multiplication law
$\mu$, one can associate a color Lie (super)algebra with
commutation factor $N$ denoted ${\cal A}_N$ by means of

$$
\l X_1, X_2 \r = \mu(X_1,X_2) - N(a,b) \mu (X_2,X_1), \ \mbox{~for
any~} (X,Y) \in {\cal A}_a \times {\cal A}_b.
$$

\noi On can easily see that the Jacobi identities are a
consequence of the associativity of the product $\mu$.
\end{remark}

Introducing a graded basis $\left\{T^{(a)}_i, i = 1,\cdots,
\text{dim }\g_a\right\}$ of $\g_a, a \in \Gamma$, the commutator
is expressed as \beqa \l T^{(a)}_\alpha, T^{(b)}_\beta\r=
C^{(a)(b)}{}_{\alpha \beta }{}{}^\gamma T^{(a+b)}_\gamma. \eeqa

\noi The scalars $C^{(a)(b)}{}_{\alpha \beta}{}{}^\gamma$ are
called the structure constants of $\g$ over the given basis.

\begin{definition}
A representation of a color Lie (super)algebra is a   mapping
$\rho : \g \to \text{End} (V)$, where $V= \underset{a \in
\Gamma}{\oplus}V_a$ is a graded vector  space such that $$\l
\rho(X), \rho(Y) \r= \rho(X) \rho(Y) - N(a,b) \rho(Y) \rho(X),$$
for all $X \in \g_a, Y \in \g_b$.
\end{definition}

We observe that for all $a,b \in \Gamma$ we  have $\rho(\g_a) V_b
\subseteq V_{a+b}$, which implies that any $V_a$ has the structure
of a $\g_0-$module. Fixing an element $v \in V$ and denoting by
$v_i$ ($i\in I$) its components, we can introduce the mapping
$\text{gr} \ : \ I \to \Gamma $ defined by $i\mapsto
\text{gr}(i)$, from which we conclude that $V_a=\left\{v=(v_i,
i\in I), \text{gr}(i)=a\right\}.$ The mapping $\text{gr}$ is
called the grading map. Now, for a given matrix representation
$\rho(T^{(a)}{}_\alpha)=M^{(a)}{}_\alpha$, where $T^{(a)}{}_\alpha
\in \g_a$, the non-vanishing indices of the matrix
$M^{(a)}{}_\alpha$ are those $\left(M^{(a)}{}_\alpha\right)_i{}^j$
satisfying the equality $\text{gr}(i)-\text{gr}(j)=a.$

\begin{example}
\label{color-gl} \label{sl} \rm{ Let $m=m_1+ \cdots +m_n$,
$V=\mathbb{C}^m$ and $\Gamma=\left\{a_1,\cdots,a_n\right\}$ be an
abelian group of order $n$. Let $\gr$ be defined by

$$\begin{array}{lll}
\text{gr}(i)=a_1&&i=1,\cdots,m_1 \\
\text{gr}(i)=a_2&&i=m_1+1,\cdots,m_1+ m_2 \\
&\vdots&\\
{\text{gr}}(i)=a_k&&i=m_1+\cdots + m_{k-1}+ 1,\cdots,m_1 + \cdots + m_k \\
&\vdots& \\
{\mbox{gr}}(i)=a_n&&i=m_1+\cdots + m_{n-1}+ 1,\cdots,m.
\end{array}
$$

\noi Consider a commutation factor $N$ satisfying the previous
relations (\ref{abelian}). We construct the color algebra (H.
Green and P. Jarvis in \cite{gj})
$\mathfrak{gl}(\{m\}_{\Gamma,N})=
  \underset{a \in \Gamma}{\oplus}\g_a$ by means of its defining relations. A
 basis of   $\g_a$ is given by the $ m \times m$ complex
matrices $(E^p{}_q)_i{}^j=\delta^p{}_i \delta_q{}^j , 1 \le p,q\le
m$ with $\text{gr}(q) -\text{gr}(p)=a$. The space $\g$ is endowed
with a color Lie (super)algebra structure:
$$
\begin{array}{lll}
\l E^p{}_q,E^r{}_s\r&=&E^p{}_q E^r{}_s-N\Big(\text{gr}(q)-\gr(p),
\gr(s)-\gr(r)\Big)E^r{}_s E^p{}_q \\
 &=& \delta^r{}_q E^p{}_s - N\Big(\text{gr}(q)-\gr(p),\gr(s)-\gr(r)\Big)
\delta^p{}_s E^r{}_q.
\end{array}$$
}
\end{example}

\noi Several subalgebras of $\mathfrak{gl}(\{m\}_{\Gamma,N})$ can
also be defined using this procedure (see \cite{sch,gj}). In
particular, if $\g$ is a color Lie (super)algebra with basis
$\left\{X_\alpha, \alpha=1,\cdots, \text{dim } \g\right\}$
satisfying $\l X_\alpha, X_\beta\r = C_{\alpha \beta}{}^\gamma
X_\gamma$, $\g$ can be embedded into some ${\mathfrak{
gl}}(\{m\})_{\Gamma,N}$ if we define $X_a=C_{ab}{}^c{\bar
e}^b{}_c$, where the ${\bar e}^b{}_c$ satisfy $\l {\bar e}^a{}_b,
{\bar e}^c{}_d\r = \delta^a{}_d {\bar
e}^c{}_b-N\left(\gr(b)-\gr(a),\gr(d)-\gr(c)\right) \delta^c{}_b
{\bar e}^a{}_d.$
\begin{example}
\label{tensor1} \rm{ Let $\g_0$ be an arbitrary Lie algebra and
let $ \left\{T_\alpha, \alpha=1,\cdots,\text{dim }\g_0\right\}$ be
a basis of $\g_0$ ($[T_\alpha,T_\beta]=f_{\alpha \beta}{}^\gamma
T_\gamma$). Consider ${\cal C}_n^2$, the complex algebra generated
by $e_1,e_2$ such that

$$e_1^n=e_2^n=1, \ e_1 e_2 = q e_2 e_1, \ q=
\exp{\left(\frac{2i\pi}{n}\right)}.$$

\noi This structure, called the generalised Clifford algebra, has
been studied by several authors (see \cite{gca} and references
therein). Introduce $e_{a,b}= e_1^a e_2^b$ a basis of ${\cal
C}_n^2$. It is easy to see that we have $e_{a,b} e_{c,d}=q^{-bc}
e_{a+c,b+d},$ and thus $\g= {\cal C}_n^2 \otimes \g_0$ is a color
Lie algebra for which the abelian group is $\Gamma = \mathbb Z_n
\times \mathbb Z_n$ and the commutation factor is $N((a,b),(c,d))=
q^{ad -bc}.$ Indeed, if we set $\g_{a,b}=\left\{T^{(a,b)}_\alpha=
e_{a,b}\otimes T_\alpha \right\}$ we have

$$\l T^{(a,b)}_\alpha,T^{(c,d)}_\beta\r =q^{-bc} f_{\alpha \beta}{}^\gamma
T^{(a+c,b+d)}_\gamma.$$

\noi It is a matter of a simple calculation to check the Jacobi
identities. Furthermore, it is known that ${\cal C}_n^2$ admits a
unique irreducible $n \times n$ faithful matrix  representation:

$$\rho(e_1)=\sigma_1=\begin{pmatrix}0&1&&\hdots&0 \\
                                    0&0&1  &\hdots&0\\
                                   \vdots&&&\ddots&\vdots \\
                                       0&\hdots&0&\hdots&1\\
                                     1&\hdots&0&\hdots&0
       \end{pmatrix},
  \rho(e_2)=\sigma_2=\begin{pmatrix} 1&0&\hdots&0 \\
                       0&q&\hdots&0 \\
                        \vdots&&\ddots&\vdots \\
                       0&\hdots&0&q^{n-1}
      \end{pmatrix}.
$$
If we introduce now the $n^2 \times n^2$ matrices

$$\rho_1= \sigma_1 \otimes {\bf 1},\
\rho_2= \sigma_2 \otimes \sigma_1 ,$$

\noi and  define $\rho_{(a,b)}=\rho_1^a \rho_2^b,$ together with a
$d-$dimensional matrix representation of $\g_0$ given by
$\rho(T_\alpha)=M_\alpha$, we obtain a $n^2 d$ dimensional
representation of $\g$ : $\rho(T^{(a,b)}_\alpha)=\rho_{(a,b)}
\otimes M_\alpha$.

This construction can be extended for larger abelian groups.
Indeed, starting from the generalised Clifford algebra ${\cal
C}_n^p$ generated by $e_1,\cdots,e_p$ satisfying $e_i^n=1,
i=1,\cdots,p, e_i e_j = q e_j e_1, 1 \le i<j \le p$ and defining
$\rho_{a_1,\cdots,a_p}=e_1^{a_1} \cdots e_p^{a_p}$, we obtain in
the same way  a color Lie algebra with abelian group $\mathbb
Z_n^p.$ Finally, let us mention that a
 similar construction can also be obtained in straightforward way by starting
from a Lie superalgebra. }
\end{example}

A somewhat different ansatz, which turns out to be of wide
interest in applications, refers to the realization of color Lie
(super)algebra in terms of differential operators \cite{gj}. Let
$\g$ be a color Lie (super)algebra with basis $\left\{T_\alpha,
\alpha=1,\cdots, \text{dim } \g \right\}$ and commutation
relations $\l T_\alpha,T_\beta\r=C_{\alpha \beta}{}^\gamma
T_\gamma$. Denote by $N$  the commutation factor. Assume further
that we have a $d-$dimensional matrix representation
$\rho(T_a)=M_a$, and introduce $d$ variables $\theta^i$ together
with their associated differential operators $\partial_i$. As
before, we assume that the index $i$ is of degree $\gr(i)$,
$\theta^i$ is of degree $-\gr(i)$ and $\partial_i$ of degree
$\gr(i)$ subjected to the following commutations relations

\beqa \label{oscil}
\l \theta^i,\theta^j\r &=& \theta^i \theta^j  + \epsilon
N\big(\gr(i),\gr(j)\big) \theta^j \theta^i
=0, \nonumber \\
\l \partial_i,\partial_j\r &=& \partial_i \partial_j+ \epsilon
N\Big(\gr(i),\gr(j)\Big)
\partial_j \partial_i=0,\\
\l \partial_i,\theta^j\r&=&\partial_i \theta^j + \epsilon
N\Big(\gr(i),-\gr(j)\Big) \theta^j \partial_i=\delta_i{}^j
\nonumber
\eeqa

\noi with $\epsilon=\pm 1$. A very elegant construction of
$\theta^i$ and $\partial_i$ can be found in \cite{k} in terms of
usual bosons (when $\epsilon=-1$)  or fermions (when
$\epsilon=1$). If we set ${\cal M}_a=\theta^i (M_a)_i{}^j
\partial_j$ and suppose that $\gr(a)=\gr(j) -\gr(i)$, a direct
computation gives

\beqa
\l {\cal M}_\alpha, \theta^i\r&=& \theta^j(M_\alpha)_j{}^i, \nonumber \\
\l {\cal M}_\alpha, \partial_i\r&=&- N(\gr(j)-\gr(i),\gr(i))
(M_\alpha)_i{}^j \partial_j, \\
\l {\cal M}_\alpha,{\cal M}_\beta\r&=&C_{\alpha \beta}{}^\gamma
{\cal M}_\gamma.\nonumber \eeqa

\noi This means that the variables $\theta^i$ are in the
fundamental representation of $\g_a$, while the variables
$\partial_i$ belong to the corresponding dual representation. These
two sets of variables, generalizing the usual bosonic and
fermionic algebras, play a  central role in differential
realizations of $\g_a$.
The next result shows that to a color Lie algebra, we can associate a graded
Lie algebra with the same grading group $\Gamma$. For this reason we call it
decoloration theorem.

\begin{proposition}
\label{decoloration} There is an isomorphism between color Lie
(super)algebras and graded Lie (super)algebras.
\end{proposition}

\begin{proof}
Consider a color Lie (super)algebra $\g=\underset{a \in
\Gamma}{\oplus} \g_a =\underset{a \in \Gamma}{\oplus}
\left<T^{(a)}_\alpha, \alpha=1,\cdots, \dim \g_a \right>$ with
commutation factor $N$ and grading group $\Gamma$. We also
introduce the commutation factor $N_+$ as defined previously. In
the case where $\g$ is a color Lie algebra we have $N_+=N$.
Consider now a graded algebra $G=\underset{a \in \Gamma} {\oplus}
G_a$ with $G_a = \mathbb C e_a$ and multiplication law given by

\beqa \label{ee} e_a e_b = \sigma(a,b) e_{a+b}, \eeqa

\noi such that the following constraint is satisfied \cite{sch}:

\beqa \label{sigma-prop} \sigma(a,b+c) \sigma(b,c)= \sigma(a,b)
\sigma(a+b,c), \forall a,b,c \in \Gamma. \eeqa

\noi If we suppose that $\sigma(a,b) \sigma^{-1}(b,a)=
N^{-1}_+(a,b)$ holds, then we have the equality $e_{-a} e_{-b}
-N_+^{-1}(a,b) e_{-b} e_{-a}=0$. This implies that
 $G$ is a subalgebra of the  associative
algebra defined by equations \eqref{oscil}. As a consequence,
condition \eqref{sigma-prop} is equivalent to assume the
associativity of the product in $G$. If we further suppose that
$N_+$ is a commutation factor, the additional condition

\beqa \label{mult-com} \sigma(a,b+c)
\sigma^{-1}(a,b)\sigma^{-1}(a,c)= \sigma(b+c,a)
\sigma^{-1}(b,a)\sigma^{-1}(c,a), \forall a,b,c \in \Gamma, \eeqa

\noi is satisfied. We call $\sigma$ a multiplier. In fact, it can
be easily shown that equation (\ref{mult-com}) is a consequence of
\eqref{sigma-prop}. Let us define

$$
\tilde \g= \underset{a \in \Gamma}{\oplus} \tilde \g_a =
 \underset{a \in \Gamma}{\oplus} e_{-a} \otimes \g_a.
$$

\noi We observe that all elements in $\tilde \g_a$ (for any $a \in
\Gamma$) are of degree zero. For $X \in \g_a, Y \in \g_b$ we set
$\tilde X=e_{-a} \otimes X,\; \tilde Y =e_{-b} \otimes Y \in
\tilde \g_a, \tilde \g_b$. From this, we derive the commutators

\beqa \label{graded} [\tilde X, \tilde Y]_\pm =\sigma(-a,-b)
e_{-a-b} \otimes \l X,Y \r \in \tilde \g_{a+b}. \eeqa

\noi These new brackets \eqref{graded} satisfy the Jacobi identity
(for $\tilde X=e_{-a} \otimes X, \tilde Y=e_{-b} \otimes Y, \tilde
Z=e_{-c} \otimes Z \in \tilde \g_a, \tilde \g_b,\tilde \g_c$)

$$
(-1)^{|\tilde Z||\tilde X|} [\tilde X, [\tilde Y, \tilde
Z]_\pm]_\pm+ (-1)^{|\tilde X||\tilde Y|} [\tilde Y, [\tilde Z,
\tilde X]_\pm]_\pm+ (-1)^{|\tilde Y||\tilde Z|} [\tilde Z, [\tilde
X, \tilde Y]_\pm]_\pm=0,
$$

\noi if the  $\sigma$ satisfies the following condition
\cite{sch}:

\beqa \label{mult-jac} \sigma(a,b+c) \sigma^{-1}(a,b)
\sigma^{-1}(a,c)
 \ \mbox{is invariant under cyclic permutation}, \forall a,b,c \in \Gamma.
\eeqa

\noi It turns out that \eqref{mult-com} and
\eqref{mult-jac} are equivalent to \eqref{sigma-prop}.
 This means that  the algebra $\tilde \g$ inherits the structure of a $\Gamma-$graded
Lie (super)algebra.
\end{proof}

\medskip

In \cite{sch}, a more general result was established, and a close
relationship between $\Gamma-$graded Lie (super)algebras
corresponding to different multiplication factors was established.
In fact we can even (composing the Grassmann-hull and the results
of Proposition \ref{decoloration})  associate a
$\Gamma-$graded-Lie algebra to a color Lie (super)algebra. This
decoloration theorem was  established in \cite{ba}. Let us briefly
recall the main steps of its proof.

\smallskip

Consider $\g = \underset{ a \in \Gamma}{ \oplus} \g_a$ a color Lie
(super)algebra with commutation factor $N$. Introduce also
$\Lambda=\underset{ a \in \Gamma}{\oplus} \Lambda_a$ a
$\Gamma-$graded algebra, canonically generated by the variables
$\theta^a_i, a \in \Gamma, i=1,\cdots,m_a$ satisfying

\beqa \label{theta} \theta^a_i \theta^b_j -N^{-1}(a,b) \theta^b_j
\theta^a_i=0. \eeqa

\noi Then the zero-graded part of $\g(\Lambda) = \Lambda \otimes
\g$,

$$
\g(\Lambda)_0 = \underset{a \in \Gamma} {\oplus} \Lambda_{-a}
\otimes \g_a,
$$

\noi is a Lie algebra. Indeed, for $X,Y,Z \in \g_a. \g_b, \g_c$
and $\theta, \psi, \eta \in \Lambda_{-a}, \Lambda_{-b},
\Lambda_{-c}$, it is not difficult to check that following
relations are satisfied:

\beqa
\begin{array}{ll}
1.& [\theta\otimes X, \psi\otimes Y] = \theta\psi \otimes \l X,Y
\r  \in
\Lambda_{-a-b} \otimes \g_{a+b}\\
2.&  [\theta\otimes X, \psi\otimes Y]   = -[\psi \otimes Y, \theta
\otimes X]
\\
3.&[\theta \otimes X,[\psi \otimes Y, \eta \otimes Z]]+
   [\psi \otimes Y,[\eta \otimes Z, \theta \otimes X]]+
   [\eta \otimes Z,[\theta \otimes X, \psi \otimes Y]]=0.
\end{array}
\eeqa

\noi This decoloration theorem has an interesting consequence.
Specifically, it means that one can associate a group to a color
Lie (super)algebra and that the parameters of the transformation
are related to the algebra $\Lambda$ above. This result was used
in the papers of Wills-Toro {\it et al} in the trefoil symmetry
frame \cite{cla-phys}. Finally, let us mention that this
decoloration theorem is in some sense the inverse procedure to the
one given in Example \ref{tensor1}.

\section{Lie algebras of order $F$}
Lie algebras of order $F$, introduced in \cite{flie}, correspond
to a different kind of extensions of Lie (super)algebras,
motivated by the implementation of non-trivial extensions of the
Poincar\'e algebra in QFT. This type of algebras is characterized
by an hybrid multiplication law: part of the algebra is realized
by a binary multiplication, while another part of the algebra is
realized {{\it via} an $F-$order product. More precisely,
 a Lie algebra of order $F$ is graded by the abelian group
$\Gamma=\mathbb Z_F$. The zero-graded part is a Lie algebra and an
$F-$fold symmetric product (playing the role of the anticommutator
in the case $F=2$) expresses the zero graded  part in terms of the
non-zero graded part.

\begin{definition}
\label{elementary} Let $F\in\mathbb{N}^*$. A ${\mathbb
Z}_F$-graded ${\mathbb C}-$vector space  ${\mathfrak{g}}=
{\mathfrak{g}}_0 \oplus {\mathfrak{g}}_1\oplus {\mathfrak{g}}_2
\dots \oplus {\mathfrak{g}}_{F-1}$ is called a complex Lie algebra
of order $F$ if
\begin{enumerate}
\item $\mathfrak{g}_0$ is a complex Lie algebra. \item For all $i=
1, \dots, F-1 $, $\mathfrak{g}_i$ is a representation of
$\mathfrak{g}_0$. If $X \in \g_0, Y \in \g_i$ then $[X,Y]$ denotes
the action of $X$ on $Y$ for any $i=1,\cdots,F-1$. \item  For all
$i=1,\dots,F-1$ there exists  an $F-$linear,
$\mathfrak{g}_0-$equivariant  map
 $\{ \cdots\} : {\cal S}^F\left(\mathfrak{g}_i\right)
\rightarrow \mathfrak{g}_0,$
 where  ${ \cal S}^F(\mathfrak{g}_i)$ denotes
the $F-$fold symmetric product of $\mathfrak{g}_i$. \item For all
$X_i \in \g_0$ and $Y_j \in \g_k$ the following ``Jacobi
identities'' hold:

\beqa \label{eq:jac} &&\left[\left[X_1,X_2\right],X_3\right] +
\left[\left[X_2,X_3\right],X_1\right] +
\left[\left[X_3,X_1\right],X_2\right] =0, \nonumber \\
&&\left[\left[X_1,X_2\right],Y_3\right] +
\left[\left[X_2,Y_3\right],X_1\right] +
\left[\left[Y_3,X_1\right],X_2\right]  =0, \nonumber \\
&&\left[X,\left\{Y_1,\dots,Y_F\right\}\right] = \left\{\left[X,Y_1
\right],\dots,Y_F\right\}  + \dots +
\left\{Y_1,\dots,\left[X,Y_F\right] \right\},  \nonumber \\
&&\sum\limits_{i=1}^{F+1} \left[ Y_i,\left\{Y_1,\dots, Y_{i-1},
Y_{i+1},\dots,Y_{F+1}\right\} \right] =0. \eeqa
\end{enumerate}
\end{definition}

\begin{remark}
If $F=1$, by definition $\mathfrak{g}=\mathfrak{g}_0$ and a Lie
algebra of order $1$ is a Lie algebra. If $F=2$, then $\g$ is a
Lie superalgebra. Therefore,  Lie algebras of order $F$ appear as
some kind of generalizations of Lie algebras and superalgebras.
\end{remark}

\begin{proposition}
Let
${\mathfrak{g}}={\mathfrak{g}}_0\oplus{\mathfrak{g}}_1\oplus\dots
\oplus{\mathfrak{g}}_{F-1}$ be a Lie algebra of order $F$, with
$F>1$. For any $i=1,\ldots,F-1$, the $\ZZ_F-$graded vector spaces
${\mathfrak{g}}_0\oplus{\mathfrak{g}}_i$ is  a Lie algebra of
order $F$. We  call these type of algebras \textit{elementary Lie
algebras of order $F$}.
\end{proposition}

\begin{remark}
\label{assos-flie} Let
 ${\cal A}= {\cal A}_0 \ \oplus\ {\cal A}_1 \  \oplus \ \cdots \ \oplus
{\cal A}_{F-1}$ be an associative $\mathbb Z_F-$graded algebra
 with multiplication $\mu$. One can associate a Lie
algebra of order $F$ to ${\cal A}$ as follows. For any $a_0, a'_0
\in {\cal A}_0, a_1,a_2,\cdots,a_F \in {\cal A}_i, i=1, \cdots
F-1$
 we have

\beqa [a_0,  a'_0]&=& \mu(a_0 ,a'_0)- \mu(a'_0,a_0) 
\in {\cal A}_0, \nonumber \\
\left[a_0, a_1\right]&=&  \mu(a_0,a_1)- \mu(a_1,a_0) \in {\cal A}_i,  
\nonumber \\
\left\{ a_1,a_2, \cdots, a_F\right\}&=&
 \mu(a_1,a_2, \cdots ,a_F) + \mbox{~ perm.} \
\in {\cal A}_0. \nonumber \eeqa

\noi Furthermore, one can easily see that the Jacobi identities
are  a consequence of the associativity of the product $\mu$.
Moreover, if ${\cal A}$ is an associative algebra and ${\cal
C}_F^1$ the commutative $F-$dimensional algebra generated by a
primitive element $e$ such that $e^F=1$, the algebra ${\cal C}_F^1
\otimes {\cal A}= (1 \otimes {\cal A}) \ \oplus \ (e \otimes {\cal
A}) \ \oplus \ \cdots \oplus \ ( e^{F-1} \otimes {\cal A})$ is
$\mathbb Z_F-$graded, and thus leads to a Lie algebra of order
$F$.
\end{remark}

\begin{definition}
A representation of an elementary Lie algebra of order $F$ is  a
linear map $\rho : ~ \g=\g_0 \oplus \g_1 \to \mathrm{End}(V)$,
such that  for all $X_i \in \g_0, Y_j \in \g_1$

\beqa \label{eq:rep}
\begin{array}{ll}
& \rho\left(\left[X_1,X_2\right]\right)= \rho(X_1) \rho(X_2)-
\rho(X_2)\rho(X_1), \cr & \rho\left(\left[X_1,Y_2\right]\right)=
\rho(X_1) \rho(Y_2)- \rho(Y_2)\rho(X_1), \cr & \rho
\left\{Y_1.\cdots,Y_F\right\}=
 \sum \limits_{\sigma \in S_F}
\rho\left(Y_{\sigma(1)}\right) \cdots
\rho\left(Y_{\sigma(F)}\right). \cr
\end{array}
\eeqa

\noindent $S_F$ being the symmetric group of $F$ elements.
\end{definition}

By construction, the vector space $V$ is graded
 $V= V_0 \oplus \cdots \oplus V_{F-1}$,
 and for all $a =\{0,\cdots, F-1\}$, $V_a$ is a $\g_0-$module. Further, the
 condition $\rho(\g_1) (V_a) \subseteq V_{a+1}$ holds.

\begin{theorem} \label{tensor} (M. Rausch de Traubenberg, M. J. Slupinski,
\cite{flie}).

\noi Let $\mathfrak{g}_{0}$ be a Lie algebra and
$\mathfrak{g}_{1}$ be a $\mathfrak{g}_{0}$-module  such that:

(i) $\g=\mathfrak{g}_{0} \oplus \mathfrak{g}_{1}$ is a Lie algebra
of order $F_1>1$;

(ii) $\mathfrak{g}_{1}$ admits  a $\mathfrak{g}_{0}$-equivariant
symmetric form of order $F_2 \ge 1$.

\noindent Then  $\g=\mathfrak{g}_{0}  \oplus \mathfrak{g}_{1}$
inherits the structure of a Lie algebra of order $F_1 +F_2$.
\end{theorem}

\noi The theorem above can be generalized to include the case
$F_1=1$ \cite{flie}.

\begin{example}
\rm{ (This is a consequence of Theorem \ref{tensor}, modified to
include $F_1=1.$) Let $\g_0$ be any Lie algebra and let $\g_1$ be
its adjoint representation. Introduce $\{J_a, a=1,\cdots,
\text{dim } \g_0\}$
 a basis of $\g_0$,
$\{A_a,a=1,\cdots, \text{ dim } \g_0\}$ the
 corresponding basis of $\g_1$ and  $g_{ab}=Tr(A_aA_b)$  the Killing form.
Then one can endow $\g=\g_0\oplus\g_1$ with a Lie algebra of order
$3$ structure given by \beqa
\left[J_a,J_b\right] &=&f_{ab}{}^c J_c, \nonumber \\
\left[J_a,A_b\right] &=&f_{ab}{}^c A_c, \nonumber \\
\left\{ A_a,A_b,A_c \right\}&=&g_{ab}J_c+g_{ac}J_b+g_{bc}J_a. \nonumber \\
\eeqa }
\end{example}

\begin{example}
\label{FP} \rm{
 Let ${\mathfrak{g}}_0 = \left< L_{\mu \nu }=-L_{\nu \mu}, P_\mu,
\ \mu,\nu=0,\cdots,D-1\right>$ be the Poincar\'e algebra in
$D-$dimensions
and ${\mathfrak{g}}_1=\left<V_\mu, \ \mu=0,\cdots, D-1 \right>$ be the $D-$%
dimensional vector representation of ${\mathfrak{g}}_0$. The
brackets
\begin{eqnarray}  \label{fp-bracket}
\left[L_{\mu \nu }, L_{\rho \sigma}\right]&=& \eta_{\nu \sigma }
L_{\rho \mu }-\eta_{\mu \sigma} L_{\rho \nu} + \eta_{\nu
\rho}L_{\mu \sigma} -\eta_{\mu
\rho} L_{\nu \sigma},  \notag \\
\left[L_{\mu \nu }, P_\rho \right]&=& \eta_{\nu \rho } P_\mu
-\eta_{\mu \rho } P_\nu, \ \left[L_{\mu \nu }, V_\rho \right]=
\eta_{\nu \rho } V_\mu
-\eta_{\mu \rho } V_\nu, \ \left[P_{\mu}, V_\nu \right]= 0,  \notag \\
\{ V_\mu, V_\nu, V_\rho \}&=& \eta_{\mu \nu } P_\rho + \eta_{\mu
\rho } P_\nu + \eta_{\rho \nu } P_\mu,  \notag
\end{eqnarray}
\noindent with the metric $\eta_{\mu
\nu}=\mathrm{{diag}(1,-1,\cdots,-1)}$ endow
${\mathfrak{g}}={\mathfrak{g}}_0 \oplus {\mathfrak{g}}_1$ with an
elementary Lie algebra of order $3$ structure which is denoted ${I\mathfrak{%
so}}_3(1,D-1)$. }
\end{example}

\begin{example}
\label{flie-gl} \rm{
\label{mat} Let $\text{mat}(m_1,m_2,m_3)$ and $\text{mat}_{\text{el}%
}(m_1,m_2,m_3)$ be the set of $(m_1 + m_2 +m_3) \times (m_1 + m_2
+ m_3)$ matrices of the form

\begin{eqnarray}  \label{gl}
\begin{array}{ll}
\text{mat}_{\mathrm{el}}(m_1,m_2,m_3) = \left\{
\begin{pmatrix}
a_0 & b_1 & 0 \\
0 & a_1 & b_2 \\
b_0 & 0 & a_2%
\end{pmatrix}
\right\}, & \text{mat}(m_1,m_2,m_3)=\left\{
\begin{pmatrix}
a_0 & b_1 & c_2 \\
c_0 & a_1 & b_2 \\
b_0 & c_1 & a_2%
\end{pmatrix}%
\right\},%
\end{array}%
\end{eqnarray}
\noindent with $a_0 \in \mathfrak{gl}(m_1), a_1 \in
\mathfrak{gl}(m_2),a_3 \in \mathfrak{gl}(m_3),$ $b_1 \in
\mathcal{M}_{m_1,m_2}(\mathbb{C}), b_2\in
\mathcal{M}_{m_2,m_3}(\mathbb{C}), b_0 \in
\mathcal{M}_{m_3,m_1}(\mathbb{C})
$, and $c_0 \in \mathcal{M}_{m_2,m_1}(\mathbb{C}), c_1\in \mathcal{M}%
_{m_3,m_2}(\mathbb{C}), c_2 \in
\mathcal{M}_{m_1,m_3}(\mathbb{C})$. A basis of this set of
matrices can be constructed as follows. Consider the $(m_1 + m_2 +
m_3)^2$ canonical matrices $e_{I}{}^{J}, \ 1 \le I, J \le m_1 +m_2
+m_3 $. With the following convention for the indices $1 \le i,j
\le m_1, m_1 + 1 \le i^{\prime },j^{\prime }\le m_1+m_2, m_1 + m_2
+ 1 \le i^{\prime \prime },j^{\prime \prime }\le m_1+m_2 + m_3$,
the generators are given by

\begin{equation*}
\begin{array}{lll}
e_{i}{}^{j}\text{ for } \mathfrak{gl}(m_1), & e_{i^{\prime }}{}^{j^{\prime }}%
\text{ for } \mathfrak{gl}(m_2), & e_{i^{\prime \prime
}}{}^{j^{\prime
\prime }}\text{ for } \mathfrak{gl}(m_3), \\
e_{i}{}^{j^{\prime }} \text{ for }
\mathcal{M}_{m_1,m_2}(\mathbb{C}), &
e_{i^{\prime }}{}^{j^{\prime \prime }} \text{ for } \mathcal{M}_{m_2,m_3}(%
\mathbb{C}), & e_{i^{\prime \prime }}{}^{j} \text{ for } \mathcal{M}%
_{m_3,m_1}(\mathbb{C}), \\
e_{i^{\prime }}{}^{j} \text{ for }
\mathcal{M}_{m_2,m_1}(\mathbb{C}), &
e_{i^{\prime \prime }}{}^{j^{\prime }} \text{ for } \mathcal{M}_{m_3,m_2}(%
\mathbb{C}), & e_{i}{}^{j^{\prime \prime }} \text{ for } \mathcal{M}%
_{m_1,m_3}(\mathbb{C}). \\
&  &
\end{array}%
\end{equation*}

\noindent Writing  $\text{mat}(m_1,m_2,m_3) =
\text{mat}(m_1,m_2,m_3)_0
\oplus \text{mat}(m_1,m_2,m_3)_1 \oplus \text{mat}(m_1,m_2,m_3)_2$ and $%
\text{mat}_{\text{el}}(m_1,m_2,m_3) =
\text{mat}_{\text{el}}(m_1,m_2,m_3)_0
\oplus \text{mat}_{\text{el}}(m_1,m_2,m_3)_1$, we denote generically by $%
X_I{}^J$ the canonical generators of degree zero, $Y_I{}^J$ the
canonical generators of degree one, and $Z_I{}^J$ those of degree
two. With these conventions, the brackets read

\begin{eqnarray}  \label{gl3}
[X_{I}{}^{J}, X_{K}{}^{L}]&=& \delta^J{}_K X_I{}^L-\delta^L{}_I
X_K{}^J,
\notag \\
\left[X_{I}{}^{J}, Y_{K}{}^{L}\right]&=& \delta^J{}_K Y_I{}^L
-\delta^L{}_I
Y_K{}^J  \notag \\
\left[X_{I}{}^{J}, Z_{K}{}^{L}\right]&=& \delta^J{}_K Z_I{}^L
-\delta^L{}_I
Z_K{}^J,  \notag \\
\left\{Y_I{}^J,Y_K{}^L,Y_M{}^N\right\}&=& \delta^J{}_K
\delta^L{}_M X_I{}^N
+ \delta^N{}_I \delta^J{}_K X_M{}^L + \delta^L{}_M \delta^N{}_I X_K{}^J \\
&+& \delta^J{}_M \delta^N{}_K X_I{}^L + \delta^N{}_K \delta^L{}_I
X_M{}^J +
\delta^L{}_I \delta^J{}_M X_K{}^N,  \notag \\
\left\{Z_I{}^J,Z_K{}^L,Z_M{}^N\right\}&=& \delta^J{}_K
\delta^L{}_M X_I{}^N + \delta^N{}_I \delta^J{}_K X_M{}^L +
\delta^L{}_M \delta^N{}_I X_K{}^J
\notag \\
&+& \delta^J{}_M \delta^N{}_K X_I{}^L + \delta^N{}_K \delta^L{}_I
X_M{}^J + \delta^L{}_I \delta^J{}_M X_K{}^N.  \notag
\end{eqnarray}
This shows that $\text{mat}(m_1,m_2,m_3)$ (resp. $\text{mat}_{\text{el}%
}(m_1,m_2,m_3)$) is endowed with the structure of Lie algebra of
order three (resp. a structure of an elementary Lie algebra of
order three). In particular, when $m_1=m_2=m_3$ the algebra above
can be rewritten as $\text{mat}(m,m,m)= {\cal C}_3^1 \otimes
\mathfrak{gl}(m)$, with $e=\begin{pmatrix} 0&1&0 \\ 0&0&1
\\ 1&0&0
\end{pmatrix}$ being a faithful matrix representation of the
canonical generator of ${\cal C}_n^1$. }
\end{example}

The question to find appropriate variables to represent Lie
algebras of order $F$ is much more involved than for color
algebras. However, in some specific cases, we were able to find
appropriate variables (see N. Mohammedi, G. Moultaka and M. Rausch
de Traubenberg in \cite{flie-phys}), and it turns out that these
variables are strongly related to Clifford algebras of polynomials
\cite{cliff}. We will give another realization below.

\section{Color Lie algebras of order $F$}
\label{fcolor} Color Lie (super)algebras of order $F$ can be seen
as a synthesis of the two types of algebras introduced previously.
Indeed, for such algebras, we have simultaneously a binary product
associated with a commutation factor and an $F-$order product. The
latter is no more fully symmetric, but is also associated with the
commutation factor. In this section we focus on color Lie
(super)algebra of order 3.

\begin{definition}
Let $\Gamma$ be an abelian group and $N$ be a commutation factor,
$\g=\g_0 \oplus \g_1$    is an elementary  color Lie
(super)algebra of order 3 if

\begin{enumerate}
\item $\g_0=  \underset{a \in \Gamma}{\oplus}\g_{0,a}$ is a color
Lie (super)algebra.
 \item $\g_1= \underset{a \in
\Gamma}{\oplus}\g_{1,a}$ is a representation of $\g_0$. If $X \in
\g_0, Y \in \g_1$ are homogeneous elements, then $\l X,Y \r$
denotes the action of $X$ on $Y$.
 \item
There exists a $\g_0-$equivariant map $\lb .,.,. \rb : \g_1
\otimes \g_1 \otimes \g_1 \to \g_0$ such that for all $Y_1 \in
\g_{1,a},Y_2 \in \g_{1,b}, Y_3 \in \g_{1,c}$ we have $\lb
Y_1,Y_2,Y_3 \rb=N(a,b) \lb Y_2,Y_1,Y_3 \rb= N(b,c)\lb
Y_a,Y_c,Y_b\rb.$
 \item The following ``Jacobi identities'' hold:

\beqa \label{eq:jaccol}
\begin{array}{l}
\l X_1,\l X_2,X_3 \r \r = \l \l X_1, X_2 \r, X_3 \r +
 N(a,b) \l X_2, \l X_1 ,X_3 \r \r, \\
\forall (X_1,X_2,X_3) \in  \g_{0,a}\times
 \g_{0,b} \times \g_{0,c}, \\
\l X_1,\l X_2,Y_3 \r \r = \l \l X_1, X_2 \r, Y_3 \r +
 N(a,b) \l X_2, \l X_1 ,Y_3 \r \r,
\\
 \forall  X_1 \in \g_{0,a},
X_2 \in \g_{0,b},  Y_3 \in \g_{1,c}, \\
\l X, \lb Y_1,Y_2,Y_3 \rb \r =
\lb \l X_1,Y_1 \r, Y_2,Y_3 \rb, \\
+ N(a,b) \lb Y_1, \l X_1,Y_2 \r,Y_3 \rb + N(a,b+c) \lb Y_1,Y_2,\l
X_1, Y_3 \r \rb, \\
 \forall X \in \g_{0,a}, Y_1 \in \g_{1,b},Y_2 \in \g_{1,c}, Y_3 \in
\g_{1,d}
\\
%
0= \l Y_1,\lb Y_2,Y_3,Y_4 \rb \r + N(a,b+c+d) \l Y_2,\lb
Y_3,Y_4,Y_1 \r \rb +
\\
N(a,b+c+d) N(b,a+c+d) \l Y_3,\lb Y_4,Y_1,Y_2 \rb \r +\\
N(a,b+c+d) N(b,a+c+d) N(c,a+b+d) \l Y_4, \lb Y_1, Y_2, Y_3 \rb \r,\\
\forall Y_1 \in \g_{1,a}, Y_2 \in \g_{1,b},Y_3 \in \g_{1,c}, Y_4 \in \g_{1,d}
\end{array}
\nonumber \eeqa
\end{enumerate}
\end{definition}

We observe that if $\Gamma = \Gamma_0 + \Gamma_1$ is a
decomposition of $\Gamma$ with respect to its $\ZZ_2-$grading, as
seen in Section \ref{color}, and such that $\Gamma_1 = 0$ (resp.
$\Gamma_1 \ne 0$),  then $\g$ is called a color Lie algebra (resp.
superalgebra). Moreover, if $\Gamma=\mathbb Z_2$ and $N(1,1)=-1$
hold, the algebra $\g$ is called a Lie superalgebra of order
three.

\begin{definition}
A representation of an elementary color Lie (super)algebra of
order 3 is a linear map $\rho~: \g \to \text{End}(V)$ satisfying
the conditions
{
\begin{enumerate}
\item $\rho\left( \l X_1,X_2\r\right)= \rho(X_1) \rho(X_2) -N(a,b)
\rho(X_2) \rho(X_1)$, for all  $X_1 \in \g_{0,a}, X_2 \in
\g_{0,b}$; \item $\rho\left(\l X_1,Y_2\r\right)= \rho(X_1)
\rho(Y_2) -N(a,b) \rho(Y_2) \rho(X_1)$, for all  $X_1 \in
\g_{0,a}, Y_2 \in \g_{1,b}$; \item $\rho\left( \lb Y_1, Y_2, Y_3
\rb\right)= \rho(Y_1) \rho(Y_2) \rho(Y_3) + N(a,b) N(a,c)
\rho(Y_2) \rho(Y_3) \rho(Y_1) + N(b,c) N(a,c) \rho(Y_3) \rho(Y_1)
\rho(Y_2) + N(b,c) \rho(Y_1) \rho(Y_3) \rho(Y_2) + N(a,b)
\rho(Y_2) \rho(Y_1) \rho(Y_3)+ N(a,b) N(a,c) N(b,c) \rho(Y_3)
\rho(Y_2) \rho(Y_1) $, for all $Y_1 \in \g_{1,a}, Y_2 \in
\g_{1,b}, Y_3 \in \g_{1,c}.$
\end{enumerate}
}
\end{definition}

By construction, the vector space $V$ is graded and we have
$V=V_0\oplus V_1 \oplus V_2$ with $V_i=\underset{a \in
\Gamma}{\oplus}V_{i,a}$. Furthermore, each $V_{i,a}$ is a
$\g_{0,0}-$module and the inclusion relation $\rho(\g_{i,a})
V_{j,b} \subseteq V_{i+j,a+b}$ holds.\\

\begin{remark}
Let
 ${\cal A}= {\cal A}_0 \ \oplus\ {\cal A}_1 \  \oplus 
{\cal A}_{2}= \left(\underset {a \in \Gamma}{\oplus} {\cal
A}_{0,a}\right) \ \oplus \ \left(\underset {a \in \Gamma}{\oplus}
{\cal A}_{1,a}\right) \ \oplus  \left(\underset
{a \in \Gamma} {\oplus}{\cal A}_{2,a}\right) $ be an associative
$\mathbb Z_3\times \Gamma-$graded algebra
 with multiplication $\mu$. One can associate a color Lie
superalgebra of order three to ${\cal A}$ defining the products in a
similar manner as in Remarks  \ref{assos-color} and
\ref{assos-flie}. In this case, the Jacobi identities are also a
consequence of the associativity of the product $\mu$. Similarly,
if ${\cal A}$ is an associative $\Gamma-$algebra and ${\cal
C}_3^1$ the commutative three-dimensional algebra generated by a
primitive element $e$ such that $e^3=1$, the algebra ${\cal C}_3^1
\otimes {\cal A}= (1 \otimes {\cal A}) \ \oplus \ (e \otimes {\cal
A}) \ \oplus \  ( e^{2} \otimes {\cal A})$ is
associative and $\mathbb Z_3\times \Gamma-$graded, and therefore
leads to a color Lie algebra of order three.
\end{remark}

The examples of color Lie (super)algebras of order $F$ are
basically of two types: we can construct a  color Lie
(super)algebra of order $F$ from either a color Lie
(super)algebra or a Lie algebra of order $F$.

\begin{example}
\rm{ Let $\mathfrak{gl}(\{m\}_{\Gamma,N})=\underset{a \in
\Gamma}{\oplus}\g_a$ be the  color Lie (super)algebra of Example
\ref{sl} and let ${\cal C}_3^1$  be the generalised Clifford
algebra with canonical generator $e$, then
\begin{enumerate}
\item ${\cal C}_3^1 \otimes \g$ is a color Lie (super)algebra of
order 3; \item $\left<1,e\right> \otimes \g$ is an elementary
color Lie (super)algebra of order 3.
\end{enumerate}
For the second algebra, following the notations of Example
  \ref{color-gl} we denote $E^p{}_q$  a basis of
$\mathfrak{gl}(\{m\}_{\Gamma,N})$ and $X^p{}_q={\mathbf 1}\otimes
E^p{}_q$ (resp.  $Y^p{}_q=e\otimes E^p{}_q$) a basis of ${\mathbf
1} \otimes \mathfrak{gl}(\{m\}_{\Gamma,N})$ (resp. of $e \otimes
\mathfrak{gl}(\{m\}_{\Gamma,N})$). Then, the trilinear brackets
read

{\small \beqa &&\lb Y^p{}_q, Y^r{}_s,Y^t{}_u\rb= \delta_q{}^r
\delta_s{}^t X^p{}_u +\nonumber \\ &&+
N(\gr(q)-\gr(p),\gr(s)-\gr(r))\ N(\gr(q)-\gr(p),\gr(u)-\gr(t)) \
\delta_s{}^t \delta_u{}^p X^r{}_q \nonumber \\&&+
N(\gr(q)-\gr(p),\gr(u)-\gr(t)) \
N(\gr(s)-\gr(r),\gr(u)-\gr(t))\ \delta_u{}^p \delta_q{}^r X^t{}_s \nonumber \\
 & &+
N(\gr(s)-\gr(r), \gr(u)-\gr(t))\ \delta_q{}^t \delta_u{}^r X^p{}_s
+
N(\gr(q)-\gr(p),\gr(s)-\gr(r))\ \delta_s{}^p \delta_q{}^t X^r{}_u \nonumber \\
&&
+ N(\gr(q)-\gr(p),\gr(s)-\gr(r)) \
N(\gr(q)-\gr(p),\gr(u)-\gr(t)) \ N(\gr(s)-\gr(r),\gr(u)-\gr(t)) \
\delta_u{}^r \delta_s{}^p X^t{}_q. \nonumber \eeqa }
}
\end{example}

\begin{example}
\rm{ Let $\g$ be an arbitrary (elementary) Lie algebra of order 3
 and let
${\cal C}_n^2$ be  the generalized Clifford algebra with canonical
generators $e_1,e_2$, then ${\cal C}_n^2 \otimes \g$ is a color
Lie algebra of order $3$ with abelian group $\mathbb{Z}_n \times
\mathbb{Z}_n$ and commutation factor
$N((a_1,b_1),(a_2,b_2))=q^{a_1 b_2 -a_2 b_1}.$ Suppose that an
elementary Lie algebra of order 3 $\g=\g_0\oplus\g_1$ is given.
Denote $\left\{X_\alpha, \alpha=1,\cdots,
\text{dim}(\g_0)\right\}$ (resp.  $\left\{Y_m, m=1,\cdots,
\text{dim}( \g_1)\right\}$) a basis of $\g_0$ (resp. $\g_1$.) such
that

\beqa \left[X_\alpha, X_\beta\right]=f_{\alpha \beta}{}^\gamma
X_\beta, \
\left[X_\alpha, Y_m\right]=R_{\alpha m}{}^n Y_n, \
\left\{Y_m,Y_n,Y_p\right\}= Q_{mnp}{}^\alpha X_\alpha. \nonumber
\eeqa

\noi Define $\g_{0,(a,b)}=\left\{X_\alpha{}^{(a,b)}
=\rho_{(a,b)}\otimes X^\alpha\right\}$ and
$\g_{1,(a,b)}=\left\{Y_m^{(a,b)}=\rho_{(a,b)} \otimes Y^m\right\}$
(with $\rho_{(a,b)}=e_1^a e_2^b$), we thus have $\underset{(a,b)
\in \mathbb Z_n \times \mathbb Z_n}{\bigoplus}{\g_{0,(a,b)}} \
\underset{(a,b) \in \mathbb Z_n \times \mathbb
Z_n}{\bigoplus}{\g_{1,(a,b)}}$ is an elementary color Lie algebra
of order 3 with brackets:

\beqa \l X_\alpha^{(a,b)}, X_\beta^{(c,d)}\r&=&
q^{-bc}f_{\alpha \beta}{}^\gamma X_\beta^{(a+c,b+d)}, \nonumber \\
\l X_\alpha^{(a,b)}, Y_m^{(c,d)}\r &=& q^{-bc}R_{\alpha m}{}^n
Y_n^{(a+c,b+d)},
\nonumber \\
\lb Y_m^{(a,b)},Y_n^{(c,d)},Y_p^{(e,f)}\rb &=&
q^{-b(c+e)-de}Q_{mnp}{}^\alpha X_\alpha^{(a+c+e,b+d+f)}. \nonumber
\eeqa

\noi As in Example \ref{tensor1} this can be extended for $\Gamma
= \mathbb{Z}_n^N$ and for color Lie superalgebras of order three.
}
\end{example}

\begin{example}
\rm{ This example is a synthesis of Examples \ref{color-gl} and
\ref{flie-gl}. Consider three abelian groups $\Gamma_1, \Gamma_2,
\Gamma_3$   and corresponding commutation factors $N_1,N_2,N_3$.
Then we define on the group $\Gamma = \Gamma_1 \times \Gamma_2
\times \Gamma_3$ the commutation factor $N(\vec a, \vec b)=
N_1(a_1,b_1) N_2(a_2,b_2) N_3(a_3,b_3)$, with $\vec a =
(a_1,a_2,a_3) \in \Gamma$ {\it etc}. Let $m_i=m_{i,1} + \cdots +
m_{i,n_i}$ with $i=1,2,3$ be three integers and
$\mathfrak{gl}(\left\{m_1\right\}_{\Gamma_1,N_1})$,
$\mathfrak{gl}(\left\{m_2\right\}_{\Gamma_2,N_2})$,
$\mathfrak{gl}(\left\{m_3\right\}_{\Gamma_3,N_3})$ be three color
Lie (super)algebras as in Example \ref{color-gl}. Introduce now
the  matrices ${\cal M}_{m_1,m_2}(\mathbb C)$ in the fundamental
representation of $\mathfrak{gl}
(\left\{m_1\right\}_{\Gamma_1,N_1})$ and in the dual of the
fundamental
 representation of $\mathfrak{gl}(\left\{m_2\right\}_{\Gamma_2,N_2})$.
In a similar way as in Example \ref{flie-gl}, we consider the set
of matrices ${\cal M}_{m_2,m_3}(\mathbb C)$ and ${\cal
M}_{m_3,m_1}(\mathbb C)$ and
 the algebra

$$
\g=\begin{pmatrix}
 \mathfrak{gl}(\left\{m_1\right\}_{\Gamma_1,N_1})&
{\cal M}_{m_1,m_2}(\mathbb C)&0 \\
0&\mathfrak{gl}(\left\{m_2\right\}_{\Gamma_2,N_2})&
{\cal M}_{m_2,m_3}(\mathbb C) \\
{\cal M}_{m_3,m_1}(\mathbb C)&0&
\mathfrak{gl}(\left\{m_3\right\}_{\Gamma_3,N_3})
\end{pmatrix}=
\begin{pmatrix}
X_i{}^j&Y_i{}^{j'}&0\\
0&X_{i'}{}^{j'}&Y_{i'}{}^{j''}\\
Y_{i''}{}^j&0&X_{i''}{}^{j''}
\end{pmatrix}
$$

\noi with the notations of Examples \ref{flie-gl}. It is obviously
a color Lie (super)algebra of order three. The various brackets
are similar to those of Example \ref{flie-gl} and \ref{color-gl}.
We just give a few brackets for completeness:

$$
\begin{array}{lll}
\l X_i{}^j,Y_k{}^{\ell'}\r&=&\delta_k{}^j Y_i{}^{\ell'},\\
\l X_{i'}{}^{j'},Y_k{}^{\ell'}\r&=&-
N_2(\gr(i')-\gr(j'), -\gr(\ell'))\delta_{i'}{}^{\ell'} Y_k{}^{j'},\\
\lb Y_i{}^{j'}, Y_{k'}{}^{\ell''},Y_{m''}{}^n\rb&=&
\delta^{j'}{}_{k'} \delta^{\ell''}{}_{m''} X^i{}_n+
N_2(-\gr(j'),\gr(k'))N_1(\gr(i),-\gr(n))\delta^{\ell''}{}_{m''}
\delta^n{}_i
X_{k'}{}^{j'} \\
&+& N_3(-\gr(\ell'), \gr(m''))N_1(\gr(i),-\gr(n))\delta^n{}_i
\delta^{j'}{}_{k'} X_{m''}{}^{\ell''}.
\end{array}
$$
}
\end{example}

\medskip
To conclude this section, we now show that there is an analogous
of the decoloration theorem established in Section \ref{color}. As
done there, one can proceed in two different (but related) ways.
To set up the main result of this theorem, consider $\g=
\left(\underset{a \in \Gamma }{\oplus} \g_{0,a}\right)\ \oplus \
 \left(\underset{a \in \Gamma }{\oplus} \g_{1,a}\right)$ a color
Lie superalgebra with grading abelian group $\Gamma$ and
commutation factor $N$. In the second approach we directly
associate to $\g$ a Lie algebra of order three,  in the same
manner as in Section \ref{color} by considering the algebra
$\Lambda = \underset{a \in \Gamma} {\oplus} \Lambda_a$, where
$\Lambda_a$ is generated by the variables $\theta^a_i$ satisfying
equation \eqref{theta}. The algebra

$$
\g(\Lambda)_0= \left(\underset{a \in \Gamma} {\oplus}\Lambda_{-a}
\otimes \g_{0,a}\right) \ \oplus \
 \left(\underset{a \in \Gamma} {\oplus}\Lambda_{-a} \otimes
 \g_{1,a}\right)
$$

\noi is a Lie algebra of order three. This is proved in a similar
way as in Section \ref{color} and only the trilinear bracket are
slightly different. Let $X \in \g_a, Y \in \g_{b}, Z \in \g_c$ and
$\theta \in \Lambda_{-a}, \psi \in \Lambda_{-b}, \eta \in
\Lambda_{-c}$. It is not difficult to check that

$$
\left\{ \theta \otimes X, \psi \otimes Y, \eta \otimes Z \right\}=
\theta \psi \eta \otimes \lb X,Y,Z \rb \in \Lambda_{-a-b-c}
\otimes \g_{0,a+b+c}.
$$

\noi The Jacobi identities involving trilinear brackets are a
consequence of the identity $[\theta_1\otimes Y_1, \{\theta_2
\otimes Y_2, \theta_3 \otimes Y_3, \theta_4 \otimes
Y_4\}]=\theta_1 \theta_2 \theta_3 \theta_4 \l Y_1, \lb Y_2,
Y_3,Y_4\rb\r$ (for any $Y_i \in \g_{1,a_i}, \theta_i \in
\Lambda_{-a_i}, i=1,\cdots,4$) together with the associativity of
the product in $\Lambda$ and equation \eqref{theta}.
This proves that $\g(\Lambda)_0$ is a Lie algebra of order three.\\

In the first correspondence we introduce $N_+$ as in Section
\ref{color} and the variables $e_a$ as in Proposition
\ref{decoloration} satisfying equations \eqref{ee} and
\eqref{sigma-prop}. Recall that the last property ensures that the
product  is associative. Then the algebra $\tilde \g =
\left(\underset{a \in \Gamma }{\oplus} e_{-a} \otimes
\g_{0,a}\right)\ \oplus \
 \left(\underset{a \in \Gamma }{\oplus}e_{-a} \otimes \g_{1,a}\right)$
is a Lie (super)algebra of order three. The proof goes along the
same lines as in Proposition \ref{decoloration}. For the bilinear
part the proof is the same as in Proposition \ref{decoloration}.
For the cubic bracket, if we take $X \in \g_a$, $Y \in \g_b$ and
$Z \in \g_a$, a simple calculation shows, using condition
\eqref{sigma-prop}, the explicit structure of the trilinear
bracket:

{\small \beqa &&
\hskip -.5truecm
\left\{ e_{-a} \otimes X, e_{-b} \otimes Y, e_{-c}
\otimes Z \right\}_\pm= (e_{-a} \otimes X)( e_{-b} \otimes Y)(
e_{-c} \otimes Z)+
\nonumber \\\hskip -.5truecm
&&(-1)^{ |X|(|Y| +|Z|)}(e_{-b} \otimes Y)(
e_{-c} \otimes Z)( e_{-a} \otimes X)
+ (-1)^{|Z|(|X| +|Y|)}  (e_{-c} \otimes Z)( e_{-a} \otimes X)(
e_{-b} \otimes Y)
\nonumber \\\hskip -.5truecm
&&+ (-1)^{|Z| |Y|} ( e_{-a} \otimes X)(  e_{-c} \otimes Z )(e_{-b}
\otimes Y) + (-1)^{|Y||X|}  (e_{-b} \otimes Y) (  e_{-a} \otimes
X)( e_{-c} \otimes Z)
\nonumber \\\hskip -.5truecm
&&+ (-1)^{|X| |Y| + |X||Z| + |Y||Z|}  (e_{-c} \otimes Z)(e_{-b}
\otimes Y)(e_{-a} \otimes X)
\nonumber \\\hskip -.5truecm
&&=\sigma(-a,-b) \sigma(-a-b,-c)e_{a + b +c} \otimes
 \lb X,Y,X \rb \in e_{-a-b-c} \otimes
\g_{0,a+b+c}, \nonumber \eeqa }

\noi where $|X|$ denotes the degree of $X$ with respect to the
$\mathbb Z_2$ grading of $\Gamma= \Gamma_0 \oplus \Gamma_1$ {\it
etc.}. Since we have $e_a e_b - N_+^{-1}(a,b) e_b e_a =0$, and the
algebra $G$ is associative,
 there is no need to prove the Jacobi identities involving
trilinear brackets (the proof being the same as in previous
cases). This illustrates how we can associate a Lie (super)algebra
of order three to a color Lie (super)algebra of order three. These
results, taken together, can be resumed in uniform manner in the
following decoloration theorem:

\begin{proposition}
\label{decol-flie} There is an isomorphism between color Lie
(super)algebras of order three and Lie (super)algebras of order
three.
\end{proposition}

This theorem can be seen as a Grassmann-hull that replaces a Lie
superalgebra by a Lie algebra introducing Grassmann variables. It
can be further be seen as a
kind of Jordan-Wigner transformation in physical applications.\\
To finish this section let us observe the following. The
decoloration theorem above and that of Section \ref{color} seem to
indicate that color Lie (super)algebras (resp. color Lie
(super)algebras of order three) do not really constitute new
objects, since they are isomorphic to Lie algebras (resp. Lie
algebras of order three). In fact, as a consequence of these
theorems, for any representation ${\cal R}$ of a color algebra
$\g$ we can construct, by means of the procedure above, an
isomorphic representation of the associated non-color algebra. The
converse of this procedure also holds. It should however taken
into account that this property does not imply that all
representations of color (resp. non-color) algebras are obtained
from representations of the corresponding non-color (resp. color)
algebras.\footnote{In particular, this way to construct
representations does not preserve dimensions, as follows at once
from the tensor products.}

\section{Quons and realization of color Lie (super)algebras of order 3}

Quons where conceived in particle statistics as one of the
alternatives to construct theories were either the Bose or Fermi
statistics are violated by a small amount \cite{quon2}. Although
observables related to particles subjected to this type of
intermediate statistics fail to have the usual locality
properties, their validity in nonrelativistic field theory and
free field theories obeying the TCP theorem has been shown. In
this section we prove that color Lie algebras of order 3 admitting
a finite dimensional linear representation can be realized by
means of quon algebras for the important case $q=0$. This result
is a generalisation of various properties that are well known for
the usual boson and fermion algebras.

\medskip

Let $-1\le q \le 1$ and consider the variables $a_i,a^i,\
i=1,\cdots n$. We define the $q$-mutator or quon-algebra by means
of \beq a_i a^j -q a^j a_i=\delta_i{}^j, \label{qalg}\eeq where no
relation between variables of the same type are postulated. The
(complex) quon algebra is denoted by ${\cal Q}_{n,q}(\mathbb
C)$.\footnote{The quon algebra originally introduced by Greenberg
is real, such that the Fock space is a Hilbert space.} For the two
extreme values of $q$ we recover the well known statistics. If
$q=-1$, together with the relations $a^ia^j+a^ja^i=0$ and
$a_ia_j+a_ja_i=0$, the quon algebra reduces to the fermion
algebra. For $q=1$, together with $a^ia^j-a^ja^i=0$ and
$a_ia_j-a_ja_i=0$, it reproduces the boson algebra. Therefore, the
quon algebra can be interpreted as an interpolation between Bose
and Fermi statistics.\footnote{The relations $a_i a_j - q
a_jq_i=0$ only  hold when the additional constraint $q^2=1$ is
satisfied \cite{quons}.}
\begin{lemma}\label{poly}
Let $M_1,\cdots,M_k$ be $(n\times n)$ complex matrices satisfying
a polynomial relation $P(M_1,\cdots,M_k)=0$. Then there exists $k$
elements ${\cal M}_k \in {\cal Q}_{n,0}(\mathbb C)\;
(k=1,\cdots,n)$ such that $P({\cal M}_1,\cdots,{\cal M}_n)=0$.
\end{lemma}
\begin{proof}
Given two arbitrary generators $a_i, a^i$ of ${\cal
Q}_{n,0}(\mathbb C)$, by equation \ref{qalg} we have $a_i
a^j=\delta_i{}^j$. This means in particular that the $n^2$
elements $e^i{}_j$ defined by $e^i{}_j=a^ia_j, 1 \le i,j\le n$ of
${\cal Q}_{n,0}$ satisfy the relation $e^i{}_j
e^k{}_\ell=\delta_j{}^k e^i{}_\ell$. Denoting by $E^i{}_j$ the
canonical generators of ${\cal M}_n(\mathbb C)$ (the $(n \times
n)$ complex matrices), the mapping $f: {\cal M}_n(\mathbb C)
 \to {\cal Q}_{n,o}$ defined by $f(E^i{}_j)=e^i{}_j $
turns out to be a injection. Therefore, since there is no kernel,
 the elements ${\cal
M}_k=a^i (M_k)_i{}^ja_j \in {\cal Q}_{n,0}(\mathbb C)$ have to
satisfy the same relations as the matrices $M_k$. Thus $P({\cal
M}_1,\cdots,{\cal M}_n)=0$.
\end{proof}

\bigskip
The quon algebra with $q=0$ has been studied in detail by
Greenberg, and constitutes an example of ``infinite statistics''
\cite{quon2}. It was moreover shown there that the $q=0$ operators
can be used as building blocks for representations in the general
$|q|\neq 1$ case. We next show that, under special circumstances,
color Lie algebras of order 3 naturally embed into a $q=0$ quon
algebra.
\begin{theorem}
Let $\Gamma$ be an abelian group, $N$ a commutation factor and
$\g$ be a color Lie (super)algebra of order 3. If $\g$ admits a
finite dimensional matrix representation, then $\g$ can be
realized by a quon algebra with $q=0$.
\end{theorem}

\begin{proof}
Suppose that the decomposition $\g= \g_0 \oplus \g_1=\underset{a
\in \Gamma} \oplus \g_{0,a} \underset{a \in \Gamma} \oplus
\g_{1,a}$ with respect to the abelian group $\Gamma$ is given. Let
$X_\alpha^{(a)}$ be a basis of $\g_{0,a}$ and $Y_m^{(a)}$ be a
basis of $\g_{1,a}$ such that the following relations hold: \beq
\label{fcol}
\begin{array}{l}
\l X_\alpha{}^{(a)}, X_\beta{}^{(b)}\r = C^{(a,b)}{}_{\alpha
\beta}{}^\gamma X^{(a+b)}_\gamma, \quad \l X_\alpha{}^{(a)},
Y_m{}^{(b)}\r = R^{(a,b)}{}_{\alpha m}{}^n
Y^{(a+b)}_n,  \\
\lb Y_m^{(a)}, Y_n^{(b)}, Y_p^{(c)}\rb =
Q^{(a,b,c)}{}_{mnp}{}^\alpha X_\alpha^{(a+b+c)}. \end{array} \eeq
Let $\rho$ be a $n$-dimensional representation of $\g$ and let
$M_\alpha^{(a)}=\rho(X_\alpha{}^{(a)}), \
N_m^{(a)}=\rho(Y_m{}^{(a)})$ denote the corresponding transformed
basic elements. Then the representation space $V$ on which the
matrices $M$ and $N$ act satisfies the decomposition $V=\underset{
a \in \Gamma}{\oplus} V_{0,a} \underset{ a \in \Gamma}{\oplus}
V_{1,a} \underset{ a \in \Gamma}{\oplus} V_{2,a} $. Now, since the
inclusions $M_{\alpha}^{(a)} V_{i,b} \subseteq V_{i,a+b}$,
$N_{m}^{(a)} V_{i,b} \subseteq V_{i+1,a+b}$ are satisfied, we can
find a basis of $V$ such that
$V=V_0 \oplus V_2 \oplus  V_1,$
 {\it i.e.}, with respect to the grading group $\mathbb Z_3$, the
block $V_i=\text{dim } \underset{ a \in \Gamma}{\oplus} V_{i,a}$
is of degree $i$, for $i=0,1,2$. With respect to this basis, the
matrices $M$ and $N$ can be rewritten as
$$M_\alpha{}^{(a)} =\begin{pmatrix} M_0{}_\alpha{}^{(a)}&0&0 \cr
                                   0&M_2{}_\alpha{}^{(a)}&0 \cr
                                   0&0&M_1{}_\alpha{}^{(a)}
                     \end{pmatrix}, \ \
N_m{}^{(a)} =\begin{pmatrix}0& M_{0-2}{}_m{}^{(a)}&0 \cr
                             0&0&M_{2-1}{}_\alpha{}^{(a)} \cr
                             M_{1-0}{}_\alpha{}^{(a)}&0&0
                     \end{pmatrix}.
$$
Let $n_i = \text{dim } V_i, \ i=1,2,3$, where obviously $n_0 + n_1
+n_2=n$. We denote by
 $v=\left(v_{i_{0}}, v_{i_{1}}, v_{i_{2}} \right)^{T}$
 the components of the vector $v \in V$ ($1 \le i_a \le n_a, \
 a=1,2,3$), and  the matrix elements of $M$ and $N$:
$\left(M_0{}_\alpha{}^{(a)}\right)_{i_0}{}^{j_0},
\left(M_1{}_\alpha{}^{(a)}\right)_{i_1}{}^{j_1},
\left(M_2{}_\alpha{}^{(a)}\right)_{i_2}{}^{j_2}$,
$\left(N_{0-2}{}_m{}^{(a)}\right)_{i_0}{}^{j_2},$
$\left(N_{2-1}{}_m{}^{(a)}\right)_{i_2}{}^{j_1},$
$\left(N_{1-0}{}_m{}^{(a)}\right)_{i_1}{}^{j_0}.$ From now on, we
adopt the convention that an index in the form $i_a, a=0,1,2$ is
of degree $a$ with respect to the grading group $\mathbb Z_3$.
Furthermore,  using the same notations as in Section \ref{color}
with respect to the grading group $\Gamma$, $v_{i_a}$ is of degree
$\gr(i_a)$. This in particular implies some relations for the
matrix elements of $M$ and $N$. For instance, considering the
matrix element $\left(N_{2-1}{}_m{}^{(a)}\right)_{i_2}{}^{j_1}$,
we have $a=\gr(i_2) -\gr(j_1)$, etc. Consider now three series of
quons ${\cal Q}_{n_0,0}(\mathbb C)=
\left<a_0{}_i,a_0{}^i,i=1,\cdots,n_0\right>$,\; ${\cal
Q}_{n_1,0}(\mathbb C)=
\left<a_1{}_i,a_1{}^i,i=1,\cdots,n_1\right>$ and ${\cal
Q}_{n_2,0}(\mathbb C)=
\left<a_2{}_i,a_2{}^i,i=1,\cdots,n_2\right>$ such that for any $m
 \ne n $, the relation  $a_m{}_i a_n{}^j=0$ holds. It follows from
the grading group $\mathbb Z_3$ that $a_a{}^{i_a}$ (respectively
$a_a{}_{i_a}$) is  of degree $a$ (resp. of degree $-a$), while,
with respect to the group $\Gamma$, $a_a{}^{i_a}$ (resp.
$a_a{}_{i_a}$) is  of degree $\gr(i_a)$ (resp. $-\gr(i_a)$). We
now define \beqa {\cal M}_\alpha{}^{(a)} &=&
\begin{pmatrix} a_0{}^{i_0} & a_2{}^{i_2}& a_1{}^{i_1} \end{pmatrix}
\begin{pmatrix}
\left(M_0{}_\alpha{}^{(a)}\right)_{i_0}{}^{j_0} &0&0 \cr
0&\left(M_2{}_\alpha{}^{(a)}\right)_{i_2}{}^{j_2} &0 \cr
0&0&\left(M_1{}_\alpha{}^{(a)}\right)_{i_1}{}^{j_1}  \cr
                     \end{pmatrix}
\begin{pmatrix} a_0{}_{j_0} \\ a_2{}_{j_2} \\ a_1{}_{j_1} \end{pmatrix}
, \nonumber \\ \\
{\cal N}_m{}^{(a)} &=&
\begin{pmatrix} a_0{}^{i_0} & a_2{}^{i_2}& a_1{}^{i_1} \end{pmatrix}
\begin{pmatrix}
0& \left(M_{0-2}{}_m{}^{(a)}\right)_{i_0}{}^{j_2}&0 \cr 0&
0&\left(M_{2-1}{}_m{}^{(a)}\right)_{i_2}{}^{j_1} \cr
\left(M_{1-0}{}_m{}^{(a)}\right)_{i_1}{}^{j_0}&0&0 \cr
                     \end{pmatrix}
\begin{pmatrix} a_0{}_{j_0} \\ a_2{}_{j_2} \\ a_1{}_{j_1} \end{pmatrix}.
\nonumber
\eeqa By definition, the matrices $M_\alpha{}^{(a)}$ and
$N_m{}^{(a)}$ satisfy the relations \eqref{fcol}. Now, applying
Lemma \ref{poly}, the elements
 ${\cal M}_\alpha{}^{(a)}, {\cal N}_m{}^{(a)} \in
{\cal Q}_{n_0,0}(\mathbb C) \oplus {\cal Q}_{n_1,0}(\mathbb C)
\oplus {\cal Q}_{n_2,0}(\mathbb C)$ satisfy the same relations.
Furthermore, since the quon algebra is an associative algebra, the
Jacobi identities are automatically satisfied.
Therefore the color algebra $\g$ has been realized in the quon
algebra ${\cal Q}_{n_0,0}(\mathbb C)
 \oplus {\cal Q}_{n_1,0}(\mathbb C) \oplus {\cal Q}_{n_2,0}(\mathbb C)$,
 finishing the proof.
\end{proof}
\medskip
It should be observed that certain types of Lie algebras of order
$3$ do not admit finite dimensional matrix representations.
However, these can realized by means of Clifford algebras of
polynomials \cite{flie-phys,cliff}. Moreover, a similar
argumentation allows to realize any given type of algebra
admitting finite dimensional representations by an appropriate set
of quons with $q=0$.

\section*{Acknowledgments}
During the preparation of this work, the author (RCS) was
financially supported by the research projects MTM2006-09152
(M.E.C.) and CCG07-UCM/ESP-2922 (U.C.M.-C.A.M.).

\end{document}